\documentclass[12pt, draftclsnofoot, onecolumn]{IEEEtran}
\usepackage{bm,cite,algorithmic,float,amsmath,amssymb}
\usepackage{enumitem}
\usepackage[linesnumbered]{algorithm2e}
\usepackage{cleveref}
\usepackage{amssymb}
\usepackage{amsmath}
\usepackage{cite}
\usepackage{amsthm}
\usepackage{framed}
\usepackage{url}
\newtheorem{theorem}{Theorem}

\newtheorem{exmp}{Example}

\usepackage{cancel}
\IEEEoverridecommandlockouts

\usepackage{graphicx,epstopdf}
\usepackage{epsfig}	
\usepackage{amsfonts}
\usepackage{tikz}
\usepackage{bbm}
\usepackage{todonotes}
\floatname{algorithm}{Algorithm}

\usepackage{ulem}
\usepackage[xcolor,pdftex]{changebar}

\begin{document}

\title{Exploiting Device Heterogeneity in Grant-Free Random Access: A Data-Driven Approach\thanks{A. Jeannerot, M. Egan, L. Chetot, and J.-M. Gorce are with Univ Lyon, Inria, INSA-Lyon, CITI, Villeurbanne, France. This work was presented in part in IEEE VTC-Spring 2023 \cite{vtc}. This work was funded in part by the ANR grant IADoc@UDL at the University of Lyon under grant number ANR-20-THIA-0007. Code used for the simulations and plotting the figures is available at: https://gitlab.inria.fr/maracas/publications/exploiting-device-heterogeneity-in-gfra}}

\author{Alix Jeannerot, Malcolm Egan and Jean-Marie Gorce}

\maketitle

\begin{abstract}
    Grant-free random access (GFRA) is now a popular protocol for large-scale wireless multiple access systems in order to reduce control signaling. Resource allocation in GFRA can be viewed as a form of frame slotted ALOHA, where a ubiquitous design assumption is device homogeneity. In particular, the probability that a device seeks to transmit data is common to all devices. Recently, there has been an interest in designing frame slotted ALOHA algorithms for networks with heterogeneous activity probabilities. These works have established that the throughput can be significantly improved over the standard uniform allocation. However, the algorithms for optimizing the probability a device accesses each slot require perfect knowledge of the active devices within each frame. In practice, this assumption is limiting as device identification algorithms in GFRA rarely provide activity estimates with zero errors. In this paper, we develop a new algorithm based on stochastic gradient descent for optimizing slot allocation probabilities in the presence of activity estimation errors. Our algorithm exploits importance weighted bias mitigation for stochastic gradient estimates, which is shown to provably converge to a stationary point of the throughput optimization problem. In moderate size systems, our simulations show that the performance of our algorithm depends on the type of error distribution. We study symmetric bit flipping, asymmetric bit flipping and errors resulting from a generalized approximate message passing (GAMP) algorithm. In these scenarios, we observe gains up to 40\%, 66\%, and 19\%, respectively.  

\end{abstract}

\begin{IEEEkeywords}
Grant-Free Random Access, Slotted ALOHA, Heterogeneous Activity, Stochastic Optimization  
\end{IEEEkeywords}

\maketitle

\section{Introduction}
A challenge in large-scale multiple access is ensuring reliable transmissions while efficiently utilizing resources and limiting latency. One approach to this challenge in cellular systems is grant-free random access \cite{Shahab2020grant,Liu2018sparse,Choi2021grant,Ding2017survey}, where active devices transmit a preamble immediately followed by data transmission. In contrast with the random access channel (RACH) procedure, grant-free random access does not require a response from a base station before transmitting a data packet. As a consequence, the access delay is reduced.

While the RACH procedure provides devices with reserved resources, this is not the case for grant-free random access. As such, devices must select their own resources, such as the time slot they will use for data transmission. Contention resolution is therefore a critical problem in grant-free random access, requiring a careful selection of slots by devices in order to reduce re-transmissions. 

Due to the lack of coordination in grant-free random access and limited information about the statistics of the network, slot selection policies are often based on variants of frame slotted ALOHA \cite{robertsALOHAPacketSystem1975}. In the original form of frame slotted ALOHA \cite{wieselthierExactAnalysisPerformance1989b}, active devices randomly select a single slot within a frame in a uniform fashion. In principle, frame slotted ALOHA reduces the chance that two active devices simultaneously access the same slot. 

Other variants of frame slotted ALOHA have also been proposed where devices may utilize more than one slot in a single frame including irregular repetition slotted ALOHA (IRSA) \cite{livaGraphBasedAnalysisOptimization2011}  and coded slotted ALOHA (CSA) \cite{paoliniCodedSlottedALOHA2015}. Nevertheless, slotted ALOHA remains the \textit{de facto} MAC protocol to coordinate massive access in IoT applications, such as machine-to-machine communications \cite{Yu2020stabilizing,Yue2023age}.

\subsection{Related Work}

A ubiquitous assumption in the design of random access protocols is that devices are homogeneous; namely, each device is active independently within a frame with a common probability. This assumption applies to the classical frame slotted ALOHA protocols \cite{wieselthierExactAnalysisPerformance1989b}, as well as recent work focusing on stability \cite{Yu2020stabilizing} and age-of-information (AoI) minimization \cite{Yue2023age,Huang2023age}.

The homogeneity assumption is also the basis of variants of frame slotted ALOHA. In the contention resolution diversity slotted ALOHA (CRDSA) scheme \cite{casiniContentionResolutionDiversity2007}, two identical copies of a data packet are transmitted in two different slots. Due to the repetitions, the reception of a single packet is sufficient to decode the data. The performance can be improved further by utilizing successive interference cancellation (SIC) \cite{verdu_multiuser_1998}, where successfully decoded packets can be subtracted from other slots allowing for additional packets from other devices to be decoded. The IRSA scheme \cite{livaGraphBasedAnalysisOptimization2011} improves upon CRDSA by allowing for additional repetitions, with the quantity selected from a probability distribution known to the base station. Decoding of data packets is achieved via message passing algorithms analogous to decoding of low-density parity-check codes (LDPC)\cite{gallagerLowDensityParityCheckCodes1962}. In recent work, the impact of SIC has been investigated \cite{Shieh2022enhanced} and tailored user identification algorithms \cite{Srivatsa2022user}.

CSA \cite{paoliniCodedSlottedALOHA2015} generalizes IRSA by splitting a data block into several packets and applying a packet-level linear block code. This yields coded packets instead of repetitions as in CRDSA and IRSA. The code for each device is selected from a predefined set according to a code probability distribution. As for IRSA, decoding at the base station is achieved via message passing algorithms. Recent work has investigated code distribution design for erasure channels \cite{Zhang2022performance}, the impact of interference cancellation errors \cite{Dumas2021design}, optimization of the power distribution \cite{Su2021noma}, and the impact of time-dependence in packet arrivals for individual devices \cite{Sousa2023study}.

 
In practice, the homogeneity assumption may not hold, either due to heterogeneous activity probabilities or statistical dependence in the activity of multiple devices. For example, devices may be sensors that observe different phenomena and have heterogeneous activity probabilities. Sensors may also observe a common phenomena, which induces correlation in their activity. 

As the homogeneity assumption is ubiquitous, a key question is whether any heterogeneity in device activities can be exploited in order to improve the performance of slotted ALOHA schemes. As frame slotted ALOHA is the basis of modern multiple access systems in the context of the IoT \cite{Yu2020stabilizing}, it is natural to first relax the homogeneity assumption for this family of protocols. 

To this end, in \cite{kalorRandomAccessSchemes2018a}, a new variant of frame slotted ALOHA has been proposed in order to account for heterogeneity in the activity of devices. In this scheme, the probability that a given active device accesses each slot is optimized based on the joint probability distribution of the device activities. The main benefit of this approach is that devices with high activity probabilities or correlation can be allocated in different slots. As a consequence, the probability of contention can be significantly reduced, leading to an improved throughput.

Two key difficulties arising in the approach proposed in \cite{kalorRandomAccessSchemes2018a} are: (i) obtaining knowledge of the activity distribution; and (ii) optimizing the slot allocation probabilities for each device. In \cite{kalorRandomAccessSchemes2018a}, these difficulties were addressed by assuming known pairwise correlations between device activities and utilizing a heuristic method to optimize the allocation probabilities based on upper and lower bounds on the expected throughput. 

In \cite{zhengStochasticResourceOptimization2021}, the two key difficulties (i) and (ii) were addressed via a data-driven approach. In each frame, the activity of each device, rather than the probability distribution, was assumed to be known. The slot allocation probabilities were then updated frame after frame via stochastic gradient ascent. In \cite{zhengStochasticResourceAllocation2022,zhengStochasticResourceOptimization2021}, contention was further managed by exploiting SIC at the physical layer. The slot allocation probabilities were then chosen that maximize the expected sum-rate or the expected number of devices with a signal-to-interference and noise ratio (SINR) exceeding a desired threshold. 

\subsection{Main Contributions}

The schemes in \cite{kalorRandomAccessSchemes2018a,zhengStochasticResourceOptimization2021,zhengStochasticResourceAllocation2022} were shown to significantly improve the performance compared with the standard frame slotted ALOHA scheme by exploiting heterogeneity in the distribution of the device activities. However, it was assumed that either the probability distribution of device activities or the device activity in each frame is perfectly known.

In practice, preamble detection and packet decoding is rarely error-free in large-scale multiple access systems \cite{Ke2020compressive,Zou2020message}, even when device heterogeneity is accounted for \cite{Chetot2021joint,Chetot2023active}. Similarly, advanced NOMA transmissions based on blind signature classification \cite{panAIDrivenBlindSignature2022} also have classification errors. Moreover, the impact of imperfect knowledge of device activities in each frame affects system performance is poorly understood, as well as how to mitigate the impact of activity estimation errors on the optimization of slot allocation probabilities. 

In this paper, we propose a throughput maximization algorithm to optimize slot allocation probabilities in frame slotted ALOHA accounting for imperfect knowledge of device activities. Our main contributions are as follows:
\begin{enumerate}
    \item[(i)] We demonstrate that the algorithms recently proposed in \cite{zhengStochasticResourceOptimization2021,zhengStochasticResourceAllocation2022} can be highly sensitive to activity estimation errors, leading to a suboptimal throughput. This is significant as, in the absence of activity estimation errors, these algorithms achieve a high performance. Moreover, we show that the suboptimality arises from bias in gradient estimates due to estimation errors.
	\item[(ii)] To mitigate the impact of the bias arising from activity estimation errors, we exploit importance weighting in a manner analogous to sample bias correction \cite{zadroznyLearningEvaluatingClassifiers2004} and bias reduction in private synthetic data \cite{ghalebikesabiMitigatingStatisticalBias2022}. We prove that by weighting the gradient estimates in the stochastic optimization algorithm in \cite{zhengStochasticResourceOptimization2021,zhengStochasticResourceAllocation2022}, the slot allocation probabilities converge to a stationary point with probability one. 
	\item[(iii)] In practice, the weight on the gradient estimates requires the evaluation of the true activity distribution and the imperfect activity distribution induced by erroneous activity estimation. As these distributions may be difficult to obtain, we propose heuristic weights that are readily available in practical systems, leading to new stochastic optimization algorithms for the slot allocation probabilities. 
	\item[(iv)] The proposed algorithms are validated by an extensive numerical study under different type of errors. Several models of activity estimation errors are considered, including independent and symmetric or asymmetric errors and a realistic error model based on generalized approximate message passing (GAMP) \cite{Chetot2021joint} arising in pilot-based device identification. Under each model for activity estimation errors, numerical results show that our algorithm is robust to activity estimation errors. Moreover, our algorithm leads to performance improvements over the algorithms in \cite{zhengStochasticResourceOptimization2021,zhengStochasticResourceAllocation2022}, which do not account for activity estimation errors.
\end{enumerate}

\subsection{Organization of the Paper}

The paper is organised as follows. In Section~\ref{sec:model} we detail the system model. In Section~\ref{sec:stochressourecalloc}, we recall the stochastic resource optimization algorithm in \cite{zhengStochasticResourceOptimization2021,zhengStochasticResourceAllocation2022}, which does not account for activity estimation errors. In Section~\ref{sec:bias}, we study the impact of activity estimation errors and develop a new stochastic resource optimization algorithm to mitigate the errors. In Section~\ref{sec:numerical}, we present a numerical study comparing our new algorithm with existing methods that do not account for activity estimation errors. In Section~\ref{sec:conclusion}, we conclude.

\section{System Model}
\label{sec:model}
Consider a network consisting of an access point (AP) equipped with $N_A$ antennas and $N$ devices equipped with a single antenna. As in common cellular standards (e.g., NB-IoT \cite{hoglundOverview3GPPRelease2017} or NR-REDCAP \cite{veeduSmallerLowerCost5G2022}), transmissions by active devices occur in frames via orthogonal frequency division multiplexing (OFDM). Each frame consists of several time-frequency resources. In NB-IoT, a common configuration utilizes a single subcarrier \cite{hoglundOverview3GPPRelease2017}. Hence, for the remainder of this paper, we consider the scenario where each frame consists of a single subcarrier and $K$ time slots \footnote{For NB-IoT, a slot is further divided in 7 symbols \cite{miaoNarrowbandInternetThings2018}.}.

\subsection{Device Activity and Resource Selection}

Within a given frame, each of the $N$ devices are either active or inactive. The activity of device $i$ in frame $t$ is represented by the binary Bernoulli distributed random variable $X^t_i \sim \mathrm{Ber}(p_i)$, where $p_i\in[0,1]$ is the probability that user $i$ is active. That is, $X^t_i \in \{0,1\}$ with $X^t_i = 1$ if device $i$ is active and $X^t_i = 0$ otherwise. The activity vector in frame $t$ is then denoted by $\mathbf{X}^t \in \{0,1\}^N$ and we denote the vector of user activity probabilities by $\mathbf{p} = [p_1,\ldots,p_N]^T\in[0,1]^N$. To simplify notations we write in the remaining of the paper $\mathbf{X}^t\sim\mathbf{p}$ to mean $\mathbf{X}^t\sim\mathrm{Ber}(\mathbf{p})$. We make the following assumptions for $\mathbf{X}^t,~t = 1,2,\ldots$:
\begin{enumerate}
	\item[(i)] $X^t_i$ is independent of $X^t_j,~j \neq i$.
	\item[(ii)] $\mathbf{X}^t$ is independent of $\mathbf{X}^{t'},~t' \neq t,~t \in \mathbb{N}$.
    \item[(iii)] $\mathbf{p}$ is \textit{not} perfectly known to the access point. In the following sections the imperfect knowledge of $\mathbf{p}$ is denoted by $\tilde{\mathbf{p}}$.
\end{enumerate}

In grant-free random access, each active device selects the time-frequency resources utilized for data transmission without coordinating with the access point nor any other devices. In the absence of any coordination, a common policy for resource selection is the classical frame slotted ALOHA protocol \cite{wieselthierExactAnalysisPerformance1989b}. In this case, devices select time-frequency resources uniformly at random. A more general policy that has recently been considered in \cite{zhengStochasticResourceOptimization2021,zhengStochasticResourceAllocation2022} allows devices to select resources with different probabilities. In this policy, the probability that device $i$ selects each resource is given in the $i$-th row of the matrix $\mathbf{A} \in \mathbb{R}^{N\times K}_+$. That is, the elements of the matrix $\mathbf{A}$ are defined by
\begin{align}
	A_{ik} = \mathrm{Pr}(\mathrm{device}~i~\mathrm{selects~resource}~k|X_i = 1)
\end{align}  
As such, $\mathbf{A} \in \mathbb{R}^{N\times K}_+$ and $\sum_{k=1}^K A_{ik} = 1,~i = 1,\ldots,N$, which ensures that each active device transmits using exactly one resource per frame. 

\subsection{Grant-Free Random Access Protocols and Device Identification}

A generic resource selection and data transmission protocol using a pilot-based identification is given in Alg.~\ref{alg:protocol_mac}. Specific protocols vary in how device identification information is communicated to the access point. In the sequel, device identification is critical for optimization of frame slotted ALOHA with heterogeneous device activity distributions. We consider two types of protocols where device identification information is included within the preamble or in the data packet. 

\RestyleAlgo{ruled}
\begin{algorithm}[!ht]
	\caption{Grant-Free Random Access Protocol}
	\label{alg:protocol_mac}
	(Downlink) Sync signal sent by the AP to indicate beginning of the first frame and inform devices about the allocation matrix $\mathbf{A}$.\\
	\While{True}{
		(Uplink) Each active device $i$ sends their preamble $\mathbf{s}_i$ on a common control channel.\\
		(Local at devices) Each active device $i$ randomly selects the slot $k_i$ with probability $A_{ik}$.\\
		(Uplink) Each active device $i$ sends their data on the drawn slot $k_i$.\\
		(Local at AP) AP decodes data present in each slot $k$. \\
		(Local at devices) Wait for beginning of next frame.\\
	}
\end{algorithm}

\subsubsection{Device Identification in the Preamble}\label{sec:preamble}

A common strategy in grant-free random access is to assign each device $i$ with a unique preamble of length $T$, $\mathbf{s}_i \in \mathbb{C}^{T}$. In each frame, the preamble of each active device is then transmitted over the block fading control channel consisting of $T$ symbols. The output of the control channel in frame $t$, $\mathbf{Y}^t \in \mathbb{C}^{T \times N_A}$, is given by 
\begin{align}\label{eq:input_output}
	\mathbf{Y}^t = \mathbf{S}\mathbf{H}^t + \mathbf{W}^t,
\end{align} 
where $\mathbf{S} = [\mathbf{s}_1,\ldots,\mathbf{s}_N]$ is the matrix of preambles, $\mathbf{H}^t \in \mathbb{C}^{N \times N_A}$ is the matrix of channel coefficients for all devices and AP antennas, and $\mathbf{W}^t \in \mathbb{C}^{T \times N_A}$ is additive white Gaussian noise. 

Given the control channel output $\mathbf{Y}^t$, the channel coefficients $\mathbf{H}^t$ and the activities $\mathbf{X}^t$ of the devices can be estimated. This is typically achieved via approximate message passing algorithms such as the generalized approximate message passing (GAMP) algorithm in \cite{Chetot2021joint}. Active devices then transmit their data by randomly selecting a data slot based on the allocation matrix $\mathbf{A}$. Data decoding is carried out by exploiting the estimate of the channels $\hat{\mathbf{H}}^t$ obtained from the approximate message passing algorithm. As the active devices have been identified from the preambles, it is not necessary to include this information within the data packet. 

\subsubsection{Device Identification in the Data}

An alternative approach is to utilize the preambles only for the purpose of channel estimation. In this case, the protocol proceeds in the same fashion as in Sec.~\ref{sec:preamble}; however, the control channel output $\mathbf{Y}^t$ is only used to estimate the channel coefficients, $\hat{\mathbf{H}}^t$, of the active users. 

As in Sec.~\ref{sec:preamble}, active devices then transmit their data by randomly selecting a data slot based on the allocation matrix $\mathbf{A}$. The key difference is that data decoding reveals both the data of the active devices and also their identification. As a consequence, the activity estimate $\hat{\mathbf{X}}^t$ is only obtained after data decoding.

\subsection{Feedback of the Allocation Matrix $\mathbf{A}$}

The grant-free random access protocol in Alg.~\ref{alg:protocol_mac} does not require a multi-step handshake procedure as is the case for the RACH procedure. However, it is necessary to inform devices of the allocation matrix $\mathbf{A}$ in order to select the resource to be used in each frame. We assume that the activity statistics of the devices typically do not change significantly over a large number of frames. As a consequence, it is only necessary to communicate the matrix $\mathbf{A}$ to the devices every $F$ frames, where the choice $F$ accounts for downlink usage constraints. 

On the other hand, the AP obtains new information about the activity of each device every frame. This information can be utilized to update the matrix $\mathbf{A}$ every frame at the AP, as is discussed in the following section. We emphasize that the AP \textit{does not} communicate the updated allocation matrix every frame.    

\section{Stochastic Resource Allocation Problem}
\label{sec:stochressourecalloc}
The stochastic resource allocation protocol described in Alg.~\ref{alg:protocol_mac} is common in frame slotted ALOHA systems and their variants. For example, in the simplest form of ALOHA, an equal probability is assigned to each time slot for all devices. In IRSA and CSA, the slot or code allocation is also randomized \cite{livaGraphBasedAnalysisOptimization2011,paoliniCodedSlottedALOHA2015}. However, these protocols have been designed under the assumption of an equal activity probability for all devices. 

In this section, we formalize the problem of stochastic resource allocation in frame slotted ALOHA for devices with \textit{heterogeneous activity probabilities}. We present an algorithm based on \cite{zhengStochasticResourceOptimization2021} to optimize the allocation \textit{under the assumption of perfect knowledge of the activities $\mathbf{X}^t$} for the throughput objective. A new algorithm to mitigate the impact of imperfect knowledge of $\mathbf{X}^t$ due to preamble detection or data decoding errors is developed in Sec.~\ref{sec:bias}.

\subsection{Objective}

In frame slotted ALOHA, a key performance metric is the expected number of collision-free transmissions in each frame, often known in this context as the throughput. Given an allocation matrix $\mathbf{A}$ and activity vector $\mathbf{X}$, the instantaneous throughput for device $i$ is given by
\begin{align}\label{eq:T_i}
	T_i(\mathbf{A};\mathbf{X}) = \sum_{k = 1}^K X_iA_{ik} \prod_{\substack{m = 1\\m \neq i}}^N (1 - X_mA_{mk}),
\end{align}
and the total instantaneous throughput by
\begin{align}
	T^{\mathbf{X}}(\mathbf{A}) = \sum_{i=1}^N T_i(\mathbf{A};\mathbf{X}).
\end{align}
After averaging over the activity vector $\mathbf{X}$, the throughput is then
\begin{align}
    T(\mathbf{A};\mathbf{p}) = \mathbb{E}_{\mathbf{X}\sim\mathbf{p}}[T^{\mathbf{X}}(\mathbf{A})].
\end{align}
In the case that the activity of each device is independent, the throughput is simplified as
\begin{align}
     T(\mathbf{A};\mathbf{p})= \sum_{i=1}^N \sum_{k=1}^K p_iA_{ik} \prod_{\substack{m = 1\\ m \neq i}}^N (1 - p_mA_{mk}).
\end{align}
In order to compare the performance of the algorithms in scenarios with different activity probability, we define the normalized throughput as 
\begin{align*}
T^N(\mathbf{A};\mathbf{p})=\frac{1}{\sum_{i=1}^Np_i}T(\mathbf{A};\mathbf{p}).
\end{align*}

\subsection{Stochastic Optimization Problem}
An optimal allocation $\mathbf{A}^*$ is a solution to the stochastic optimization problem 
\begin{align}\label{eq:opt_prob}
	\mathbf{A}^* \in \arg\max_{\substack{\mathbf{A} \in \mathbb{R}_+^{N \times K}:\\ \sum_{k=1}^K A_{ik} = 1,~i = 1,\ldots,N}} T(\mathbf{A};\mathbf{p}).
\end{align}   
In general, the objective $T(\mathbf{A};\mathbf{p})$ is non-convex in $\mathbf{A}$. Moreover, solutions $\mathbf{A}^*$ may often admit non-binary solutions, particularly when $p_i < \frac{1}{2},~i = 1,\ldots,N$.

If the distribution $\mathbf{p}$ is known to the AP, it is in principle possible to obtain solutions to (\ref{eq:opt_prob}) via gradient descent. However, this approach has a very high complexity due to the fact that the throughput has $2^N$ terms. Indeed, 
\begin{align}
	T(\mathbf{A};\mathbf{p}) = \sum_{\mathbf{x} \in \{0,1\}^N} \mathrm{Pr}(\mathbf{X=x})T^{\mathbf{X}}(\mathbf{A}).
\end{align}

In practice, only limited knowledge of $\mathbf{p}$ is available as it must be estimated based on the activity estimates $\hat{\mathbf{X}}^t$ obtained from the preamble or data decoding. As the activity estimates are error prone, $\hat{\mathbf{X}^t}$ can be considered as being drawn from another distribution $\hat{\mathbf{p}}$. An alternative approach is to directly utilize the activity estimates $\hat{\mathbf{X}}^t$ to solve the allocation problem in (\ref{eq:opt_prob}). As observed in \cite{zhengStochasticResourceOptimization2021}, this can be achieved via projected stochastic gradient ascent. 

Suppose, for the moment, that the activity vectors $\mathbf{X}^1,\mathbf{X}^2,\ldots$ for each frame are perfectly known to the AP. The projected stochastic gradient ascent algorithm (summarized in Alg.~\ref{alg:SGA_no_bias}) updates the allocation matrix $\mathbf{A}$ each frame via the recursion 
\begin{align}
	\mathbf{A}^{t+1} = \Pi_{\mathcal{H}}\left\{\mathbf{A}^t + \gamma^{t+1}g(\mathbf{A}^t;\mathbf{X}^{t+1})\right\},
\end{align}
where $(\epsilon^t)$ is a positive step size sequence,  $\Pi_{\mathcal{H}}$ denotes the Euclidean projection on the constraint set 
\begin{align}
	\mathcal{H} = \{\mathbf{A} \in \mathbb{R}^{N \times K}_+: \sum_{k=1}^K A_{ik} = 1,~i = 1,\ldots,N\},
\end{align}
and $g(\mathbf{A}^t;\mathbf{X}^{t+1})$ is a stochastic gradient estimate based on the activity vector $\mathbf{X}^{t+1}$. In particular, the gradient estimate for the throughput objective is given by (the superscript $t$ is dropped for the shake of clarity).
\begin{align}
	g(\mathbf{A};\mathbf{X})&=\sum_{n=1}^Ng_n(\mathbf{A};\mathbf{X})\label{eq:gradient_as_sum},
\end{align}
\begin{align}
    g(\mathbf{A};\mathbf{X})_{ql}&=\sum_{n=1}^N g_n(\mathbf{A};\mathbf{X})_{ql}\nonumber\\
	&= X_q\prod_{\substack{m=1 \\ m\neq q}}^N\left(1-X_mA_{ml}\right)\notag\\
	&~~~-\sum_{\substack{n=1\\n\neq q}}^NX_qX_nA_{nl}\prod_{\substack{m=1\\m\neq n\\m\neq q}}^N(1-X_mA_{ml})
\end{align}
where
\begin{align*}
	g_n(\mathbf{A};\mathbf{X})_{ql} = 
	\begin{cases}
		X_q\prod_{\substack{m=1 \\ m\neq q}}^N\left(1-X_mA_{ml}\right)& \mathrm{if~} q=n\\
		-X_qX_nA_{nl}\prod_{\substack{m=1\\m\neq n\\m\neq q}}^N(1-X_mA_{ml})&\mathrm{if~} q\neq n
	\end{cases}
\end{align*}
Note that in this case, $g(\mathbf{A};\mathbf{X})\in\mathbb{R}^{N\times K}$ is an unbiased estimate of $\nabla_{\mathbf{A}^t} T(\mathbf{A};\mathbf{p})=\mathbb{E}_{\mathbf{X}\sim\mathbf{p}}[g(\mathbf{A};\mathbf{X})]$ due to the absence of user identification errors. As such, under appropriate conditions on the step size sequence $\{\gamma^t\}$ and perfect knowledge of $\mathbf{X}^1,\mathbf{X}^2,\ldots$, the iterates $\{\mathbf{A}^t\}$ converge almost surely to a stationary point of $T(\cdot)$ \cite{zhengStochasticResourceOptimization2021}.

\begin{algorithm}[!ht]
	\caption{Stochastic optimization algorithm when user identification is error-free.}
	\label{alg:SGA_no_bias}
	Choose initial allocation matrix $\mathbf{A}^1 \in \mathbb{R}^{N \times K}$ such that $\sum_{k=1}^K A_{n,k}^1 = 1,~n = 1,\ldots,N$, and step-size sequence $\{\gamma^t\}$ with $\gamma^t > 0,~t = 1,2,\ldots$\\
	$t \leftarrow 1$.\\
	\While {not converged}{
        Based on $\mathbf{X}^{t+1}\sim\mathbf{p}$, compute an unbiased estimate $g(\mathbf{A}^t;\mathbf{X}^{t+1})$ of $\nabla_{\mathbf{A}^t} T(\mathbf{A};\mathbf{p})$\\
		$\mathbf{A}^{t+1}\leftarrow\Pi_\mathcal{H}[\mathbf{A}^t+\gamma^{t+1}g(\mathbf{A}^t;\mathbf{X}^{t+1})]$\\
		$t\leftarrow t+1$}
\end{algorithm}


In practice, however, perfect knowledge of $\mathbf{X}^1,\mathbf{X}^2,\ldots$ is not available to the AP due to errors in device identification. As we discuss in the following section, this can lead to a significant degradation in performance. 

\section{Mitigating Errors in Activity Estimation}
\label{sec:bias} 

\subsection{Impact of Activity Estimation Error}

As noted in the previous section, perfect knowledge of device activities $\mathbf{X}^1,\mathbf{X}^2,\ldots$ is not available to the AP. A key question is then the impact on the performance of Alg.~\ref{alg:SGA_no_bias}. In fact, errors in the activity estimates can lead to a significant performance degradation, even in small networks as illustrated in the following example.

\begin{exmp}
	\label{exmp:userdetection_errors}
	Consider a network consisting of $3$ devices sharing two time slots. The true activity probabilities of the devices are $\mathbf{p}=\begin{bmatrix}0.3&0.4&0.9\end{bmatrix}$. Suppose that the user identification algorithm is not always able to distinguish between the first and last users. In particular, with probability $\epsilon$, the estimated user activity vector available to the AP is in fact governed by $\hat{\mathbf{p}}=\begin{bmatrix}0.9&0.4&0.3\end{bmatrix}$. The observed distribution is thus $\mathbf{p}^\prime=(1-\epsilon)\mathbf{p}+\epsilon\hat{\mathbf{p}}$.
	
	Fig.~\ref{fig:exmp} shows the impact of $\epsilon$ on the network throughput $T(\mathbf{A})$ after optimizing $\mathbf{A}$ with Alg.~\ref{alg:SGA_no_bias} for 500 frames. The throughput of standard frame slotted ALOHA is also plotted, where every element of the allocation matrix $\mathbf{A}$ is $\frac{1}{2}$). The case $\epsilon=0$ corresponds to the case of perfect user identification and $\epsilon=1$ to the case where the activity vector is always drawn from $\hat{\mathbf{p}}$. 
	
	Observe that as $\epsilon$ increases (corresponding to a higher probability of errors in the estimation of $\mathbf{X}$), the throughput significantly decreases compared with the baseline where $\mathbf{X}$ has been perfectly estimated. If the average number of errors is too large, then the resulting throughput is worse than ALOHA.
	
        \begin{figure}[h!]
		\centering
		\includegraphics[width=\linewidth]{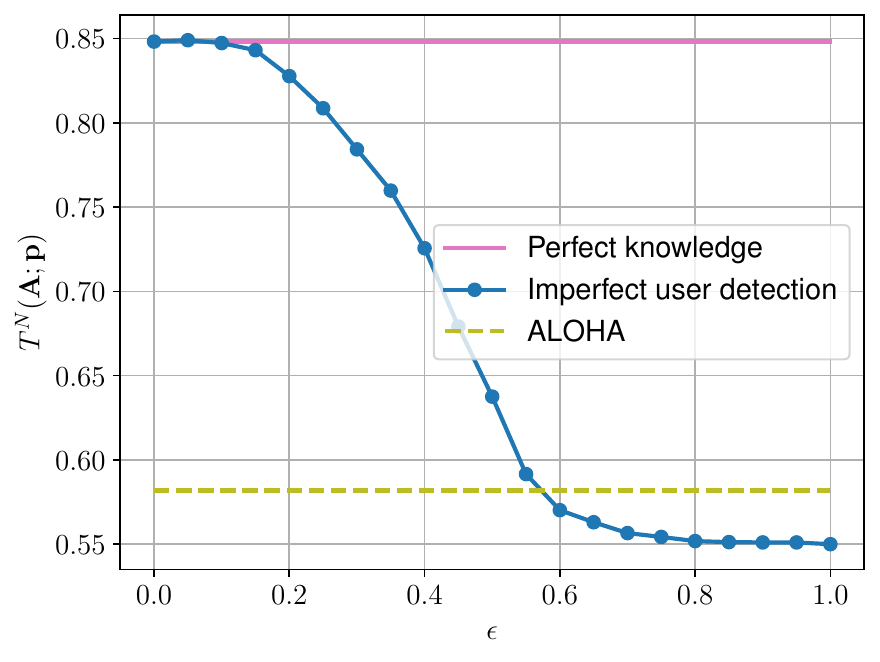}
		\caption{Expected throughput with user identification errors for a network of $3$ devices sharing two slots. Baseline, in blue, corresponds to $\epsilon=0$. Drawing samples $\mathbf{X}_t$ with probability $\epsilon$ from $\hat{\mathbf{p}}$ instead of $\mathbf{p}$ (in orange) can reduce the throughput by up to $35\%$.}
		\label{fig:exmp} 
	\end{figure}
\end{exmp}

\subsection{Unbiasing Gradient Estimates}

In order to develop an algorithm to mitigate the impact of user identification errors, we first examine the impact of the errors on the SGA algorithm in Alg.~\ref{alg:SGA_no_bias}. To this end, consider the update rule 
\begin{align}
	\mathbf{A}^{t+1} &= \Pi\left\{\mathbf{A}^t + \gamma^{t+1}g(\mathbf{A}^t;\mathbf{X}^{t+1})\right\}\notag\\
	&= \mathbf{A}^t + \gamma^{t+1}g(\mathbf{A}^{t};\hat{\mathbf{X}}^{t+1}) + \gamma^{t+1}\mathbf{Z}^{t+1},
\end{align}
where $\gamma^{t+1}\mathbf{Z}^{t+1}$ is the smallest vector (w.r.t the Euclidean norm) needed to project $\mathbf{A}^t + \gamma^{t+1}g(\mathbf{A}^{t};\hat{\mathbf{X}}^{t+1})$ into the constraint set $\mathcal{H}$. We then have
\begin{align}
\mathbf{A}^{t+1} &= \mathbf{A}^t + \gamma^{t+1}\left(\mathbb{E}_{\mathbf{X}^t\sim{\mathbf{p}}}[g(\mathbf{A}^t;\mathbf{X}^{t+1})] + \beta^{t+1} + \partial \mathbf{M}^{t+1}\right)\notag\\
	&~~~ + \gamma^{t+1}\mathbf{Z}^{t+1},
\end{align}
where 
\begin{align}
    \beta^{t+1} &= \mathbb{E}_{\hat{\mathbf{X}}{\sim\hat{\mathbf{p}}}}[g(\mathbf{A}^t;\hat{\mathbf{X}}^{t+1})] - \mathbb{E}_{\mathbf{X}\sim\mathbf{p}}[g(\mathbf{A}^t;\mathbf{X}^{t+1})]\notag\\
	\partial \mathbf{M}^{t+1} &= g(\mathbf{A}^t;\hat{\mathbf{X}}^{t+1}) - \beta^{t+1} - \mathbb{E}_{\mathbf{X}\sim\mathbf{p}}[g(\mathbf{A}^t;\mathbf{X}^{t+1})].
\end{align}
Note that $\mathbb{E}_{t+1}[\partial \tilde{\mathbf{M}}^{t+1}] = 0$, where $\mathbb{E}_{t+1}[\cdot]$ is the conditional expectation on with respect to the $\sigma$-algebra generated by the iterates up to time $t$. 

In the case that there are no user identification errors, $\beta^{t+1} = 0$, $g(\mathbf{A}^t;\mathbf{X}^{t+1})$ is an unbiased estimate of the gradient. Moreover, Alg.~\ref{alg:SGA_no_bias} is guaranteed to converge to a stationary point under the hypotheses in Theorem~5.2.10 of \cite{kushnerStochasticApproximationRecursive2003}. However, when $\beta^{t+1} \neq 0$, the gradient estimate is biased. As a consequence, there is no guarantee of convergence, even to a stationary point. The resulting performance reduction is illustrated in Fig.~\ref{fig:exmp}.

\subsection{Proposed Algorithm}

As the cause of the performance reduction is due to biased gradient estimates, it is necessary to utilize a bias reduction method. Bias reduction algorithms have been exploited in various machine learning problems; e.g., sample selection bias correction \cite{zadroznyLearningEvaluatingClassifiers2004} and privacy \cite{ghalebikesabiMitigatingStatisticalBias2022}. We now adapt these methods to our problem. 

Consider the weighting function $w:\mathbf{x} \mapsto w(\mathbf{x})$ for $\mathbf{x} \in \{0,1\}^N$ consisting of the ratio between a target distribution and a proposal distribution, it is defined by 
\begin{align}\label{eq:weight}
	w(\mathbf{x}) = \frac{\mathrm{Pr}(\mathbf{X} = \mathbf{x})}{\mathrm{Pr}(\hat{\mathbf{X}} = \mathbf{x})},
\end{align}
where $\mathbf{X}$ is the true activity vector and $\hat{\mathbf{X}}$ is the estimated activity vector. Suppose that $w(\mathbf{x}) < \infty,~\mathbf{x} \in \{0,1\}^N$ where we adopt the convention that $\frac{0}{0} = 0$. For the purposes of bias reduction, a key property of $w(\mathbf{x})$ is then that for all $\mathbf{A} \in \mathcal{H}$, 
\begin{align}\label{eq:IS_identity}
    \mathbb{E}_{\mathbf{X}\sim\mathbf{p}}[g(\mathbf{A};\mathbf{X})] &= \sum_{\mathbf{x} \in \{0,1\}^N} g(\mathbf{A};\mathbf{x})\mathrm{Pr}(\mathbf{X} = \mathbf{x})\notag\\
	&= \sum_{\mathbf{x} \in \{0,1\}^N} g(\mathbf{A};\mathbf{x})w(\mathbf{x})\mathrm{Pr}(\hat{\mathbf{X}} = \mathbf{x})\notag\\
    &= \mathbb{E}_{\hat{\mathbf{X}}\sim\hat{\mathbf{p}}}[w(\hat{\mathbf{X}})g(\mathbf{A};\hat{\mathbf{X}})].
\end{align} 
In other words, $w(\hat{\mathbf{X}})g(\mathbf{A};\hat{\mathbf{X}})$ is an unbiased estimate of $\mathbb{E}_{\mathbf{X}\sim\mathbf{p}}[g(\mathbf{A};\mathbf{X})]$ as long as $w(\mathbf{x}) < \infty$ for all $\mathbf{x}$ such that $\mathrm{Pr}(\mathbf{X} = \mathbf{x}) > 0$. As such, as we rigorously establish in Theorem~\ref{thrm:convergence}, this choice of weight overcomes the key problem preventing convergence of Alg.~\ref{alg:SGA_no_bias}.

\begin{algorithm}[!ht]
	\caption{Stochastic optimization algorithm with user identification errors}
	\label{alg:SGA_bias_reduced}
	Choose initial allocation matrix $\mathbf{A}^1 \in \mathbb{R}^{N \times K}$ such that $\sum_{k=1}^K A_{n,k}^1 = 1,~n = 1,\ldots,N$, and step-size sequence $\{\gamma^t\}$ with $\gamma^t > 0,~t = 1,2,\ldots$\\
	$t \leftarrow 1$.\\
	\While {not converged}{
        Based on the estimate $\hat{\mathbf{X}}^{t+1}\sim\hat{\mathbf{p}}$, compute a biased estimate $g(\mathbf{A}^t;\hat{\mathbf{X}}^{t+1})$ of $\nabla_{\mathbf{A}^t} T(\mathbf{A};\mathbf{p})$\\
		$\mathbf{A}^{t+1}=\Pi_\mathcal{H}[\mathbf{A}^t+\gamma^{t+1}w( \hat{\mathbf{X}}^{t+1})g(\mathbf{A}^t;\hat{\mathbf{X}}^{t+1})]\label{stp:update_weight}$\\
		$t\leftarrow t+1$}
\end{algorithm}

Incorporating the weight $w(\mathbf{x})$ in (\ref{eq:weight}) into Alg.~\ref{alg:SGA_no_bias} leads to Alg.~\ref{alg:SGA_bias_reduced}. This algorithm has the following convergence guarantee. 

\begin{theorem}\label{thrm:convergence}
	The iterates $\mathbf{A}^t$ of Alg.~\ref{alg:SGA_bias_reduced} converge almost surely as $t \rightarrow \infty$ to a stationary point provided that the step size sequence $\{\gamma^t\}$ satisfies
	\begin{enumerate}
		\item[(i)] $\gamma^t > 0,~t > 0$;
		\item[(ii)] $\sum_{t=1}^{\infty} \gamma^t = \infty$;
		\item[(iii)] $\sum_{t=1}^{\infty} (\gamma^t)^2 < \infty$,
	\end{enumerate}
	and $w(\mathbf{x}) < \infty$ for all $\mathbf{x} \in \{0,1\}^N$. 
\end{theorem}

\begin{proof}
	See Appendix~\ref{app:proof}.
\end{proof}

\subsection{Computing the Importance Weight}
\label{sec:computing_IW}
Theorem~\ref{thrm:convergence} shows that even when there are errors in $\hat{\mathbf{X}}$, Alg.~\ref{alg:SGA_bias_reduced} converges to a stationary point for the problems in (\ref{eq:opt_prob}). This is in contrast to the standard SGA algorithm in Alg.~\ref{alg:SGA_no_bias}, for which these guarantees hold only when there are no errors. 

However, the implementation of Alg.~\ref{alg:SGA_bias_reduced} requires knowledge of both $\mathbf{p}$ and $\hat{\mathbf{p}}$ in order to compute the weight $w(\mathbf{x})$ in (\ref{eq:weight}). In practice, the AP does not have perfect knowledge of $\mathbf{p}$ nor $\hat{\mathbf{p}}$. Nevertheless, estimates can be obtained as follows:
\begin{enumerate}
    \item[(i)] Estimation of the target distribution $\mathbf{p}$:
        The target distribution corresponds to the prior distribution required by the GAMP algorithm. Knowledge of the prior is obtained via expectation-maximization algorithms (see [36] for details on GAMP and estimation of the prior). Another possibility is that device can periodically feedback estimates of activity probabilities.
        We model the resulting estimated distribution, accounting for errors, by adding a Gaussian perturbation to the true distribution, $\tilde{\mathbf{p}}=\mathbf{p}+\bm{\eta}$, with $\bm{\eta}\sim\mathcal{N}(0,\sigma^2I)$, where $\tilde{\mathbf{p}}$ is clipped, if needed, to be between 0 and 1. We denote $\tilde{\mathbf{X}}\sim\mathrm{Ber}(\tilde{\mathbf{p}})$ and $\tilde{w}(\mathbf{x}) = \frac{\mathrm{Pr}(\tilde{\mathbf{X}} = \mathbf{x})}{\mathrm{Pr}(\hat{\mathbf{X}} = \mathbf{x})}$ the estimated importance weight.
	\item[(ii)] Estimation of the proposal distribution $\hat{\mathbf{p}}$: the estimation of $\hat{\mathbf{p}}$ can be achieved by empirical distribution estimation via the outputs of the user identification algorithm (e.g., GAMP). 
\end{enumerate}
As a consequence, the weight $w(\mathbf{x})$ in (\ref{eq:weight}) is not perfectly known in practical systems. Nevertheless, as we show in the following section, an imperfect estimate of the weight $w(\mathbf{x})$ still yields improved performance over Alg.~\ref{alg:SGA_no_bias}, which ignored the impact of user identification errors. \\
Moreover, if $\mathrm{Pr}(\hat{\mathbf{X}}=\mathbf{x})$ is close to 0, $w(\mathbf{x})$ will be very large. This mean that the optimization algorithm can take very large steps, pushing the optimization towards undesired local maxima. To increase the stability of the algorithm, we propose Alg. \ref{alg:SGA_bias_reduced_imperfect_w}, in which we clip the weight of iteration $t$ to a maximal value of $\kappa^t > 0$.

\begin{algorithm}[!ht]
	\caption{Stochastic optimization algorithm with user identification errors, approximate target distribution and weight clipping.}
	\label{alg:SGA_bias_reduced_imperfect_w}
	Choose initial allocation matrix $\mathbf{A}^1 \in \mathbb{R}^{N \times K}$ such that $\sum_{k=1}^K A_{n,k}^1 = 1,~n = 1,\ldots,N$, clipping parameter sequence $\kappa^t > 0$, and step-size sequence $\{\gamma^t\}$ with $\gamma^t > 0,~t = 1,2,\ldots$\\
	$t \leftarrow 1$.\\
	\While {not converged}{
        Based on the estimate $\hat{\mathbf{X}}^{t+1}\sim\hat{\mathbf{p}}$, compute a biased estimate $g(\mathbf{A}^t;\hat{\mathbf{X}}^{t+1})$ of $\nabla_{\mathbf{A}^t} T(\mathbf{A};\mathbf{p})$\\
		$\tilde{w}(\hat{\mathbf{X}}^{t+1})=\min\left\{\kappa^t,\frac{\mathrm{Pr}(\tilde{\mathbf{X}} = \mathbf{x})}{\mathrm{Pr}(\hat{\mathbf{X}} = \mathbf{x})}\right\}$\\
		$\mathbf{A}^{t+1}=\Pi_\mathcal{H}[\mathbf{A}^t+\gamma^{t+1}\tilde{w}(\hat{\mathbf{X}}^{t+1})g(\mathbf{A}^t;\hat{\mathbf{X}}^{t+1})]$\\
		$t\leftarrow t+1$}
\end{algorithm}

\section{Numerical Results}
\label{sec:numerical}

\subsection{Parameters and Baseline Methods}
To support our theoretical analysis, we provide simulation results for different type of errors introduced by the user detection algorithms. We compare the throughput given by several methods: 
\begin{itemize}
    \item Optimization of $\mathbf{A}$ using Alg.~\ref{alg:SGA_no_bias} with perfect detection of users, as baseline, where the true activity vector $\mathbf{X}$ is known,
    \item Optimization of $\mathbf{A}$ using Alg.~\ref{alg:SGA_no_bias} with imperfect user detection representing the optimization obtained when using the error prone activity vector $\hat{\mathbf{X}}$,
    \item Optimization of $\mathbf{A}$ using Alg.~\ref{alg:SGA_bias_reduced_imperfect_w} with true weight $w$ for Fig.~\ref{fig:bsc} and Fig.~\ref{fig:bec} (with true target $\mathbf{p}$ empirically estimated proposal $\hat{\mathbf{p}}$ for Fig.~\ref{fig:gamp})
    \item Optimization of $\mathbf{A}$ using Alg.~\ref{alg:SGA_bias_reduced_imperfect_w} with an imperfect weight due to a target distribution $\tilde{\mathbf{p}}$ that is known up to a Gaussian perturbation of different standard deviations $\sigma$. 
    \item A greedy allocation $\mathbf{A}_h$ obtained via Alg.~\ref{alg:greedy}: the users that are the most likely to transmit are allocated their own slot, while all the others users share a single slot. 
    
	\begin{algorithm}[!ht]
	    \caption{Construction of Greedy Allocation Matrix $\mathbf{A}_h$}
		\label{alg:greedy}
		Given a sorted vector of activity probabilities $\mathbf{p}$ in the network and $K$ slots\\
		$k \leftarrow 1$\\
        $\mathbf{A}\leftarrow \bm{0}_{N\times K}$\\
		\For{$i\leftarrow 1;~i\leq N-K;~i\leftarrow i+1$}{
		    $A_{i,1}\leftarrow 1$
		}
		\For{$k\leftarrow 2;~k\leq K;~k\leftarrow k+1$}{
		    $A_{N-k+1,k}\leftarrow 1$
      }
	\end{algorithm}
	
    \item The Frame Slotted ALOHA allocation, a constant matrix having $A_{ik}=\frac{1}{K}$ as elements,
    \item $\mathbf{A}^0$ represents the initial random matrix that was used. 
\end{itemize}
In the following figures, we plot the normalized throughput, corresponding to the expected throughput per frame of the optimized matrix $\mathbf{A}$ for a given network of $N$ users with activity probability $p_i$. It should be noted that this policy is implemented for a large number of frames and thus a small difference in the per-frame throughput can lead to significant differences in the number of successfully decoded packets over long time periods.\\ 
We test our methods on two different models of user detection (with symmetric and asymmetric errors) and with an active user detection algorithm based on GAMP \cite{Rangan2011generalized}. For any type of user detection algorithms, the choice of the initial matrix $\mathbf{A}_0$ is difficult; the problem (\ref{eq:opt_prob}) is non-convex and typically has local maxima with poor performance. As the diameter of the constraint set $\mathcal{H}$ is large, a badly chosen initial matrix might lead to poorly performing local maxima. To overcome this issue we assume that the access point can compute the network throughput for a given matrix. As such, the AP can keep track of several initial matrices and choose to send in the downlink the best one. For our simulations, 12 initial matrices $\mathbf{A}^0$ are drawn uniformly at random in $\mathcal{H}$.\\
The simulated networks consist of $N=20$ users and $K=5$ slots, consistent with the number of slots in an NB-IoT frame for the 3.75kHz subcarrier spacing \cite{miaoNarrowbandInternetThings2018}. We run the different methods for 10000 frames. We should note that in this scenario of NB-IoT, the duration of a slot is 2ms and a frame 10ms \cite{miaoNarrowbandInternetThings2018}, thus 10000 frames represent 100 seconds. Each method is simulated using the same random sequence of activity vectors and we repeat the simulations 20 times, allowing to isolate the effect of the detection errors on the resulting allocation.   
For the weight $\tilde{w}$ in Alg.~\ref{alg:SGA_bias_reduced_imperfect_w}, a new target distribution $\tilde{\mathbf{p}}$ is drawn for each run according to the description in Section \ref{sec:computing_IW}.

The step size and the maximum weight in the different algorithms are set constant across all simulations, frames and type of channels with value $\gamma^t=\frac{1}{100}$ and $\kappa^t=5$. In our following figures, the shaded area represent $\pm$ the standard deviation. 

\subsection{Symmetric Errors}
\label{sec:num:sym}
Consider a detection algorithm where the false alarm and miss-detection probabilities for each device are equal and given to $p_{\mathrm{flip}} \in [0,0.5]$. In other words, if device $i$ is active, the probability it is not detected is $p_{\mathrm{flip}}$. Similarly, if device $i$ is not active, it is detected with probability $p_{\mathrm{flip}}$. In this case, 
\begin{align}
	\hat{p}_i = p_i + p_{\mathrm{flip}} - 2p_{\mathrm{flip}}p_i,~i = 1,\ldots,N.
\end{align}

\begin{figure}[h!]
\centering
\includegraphics[width=\linewidth]{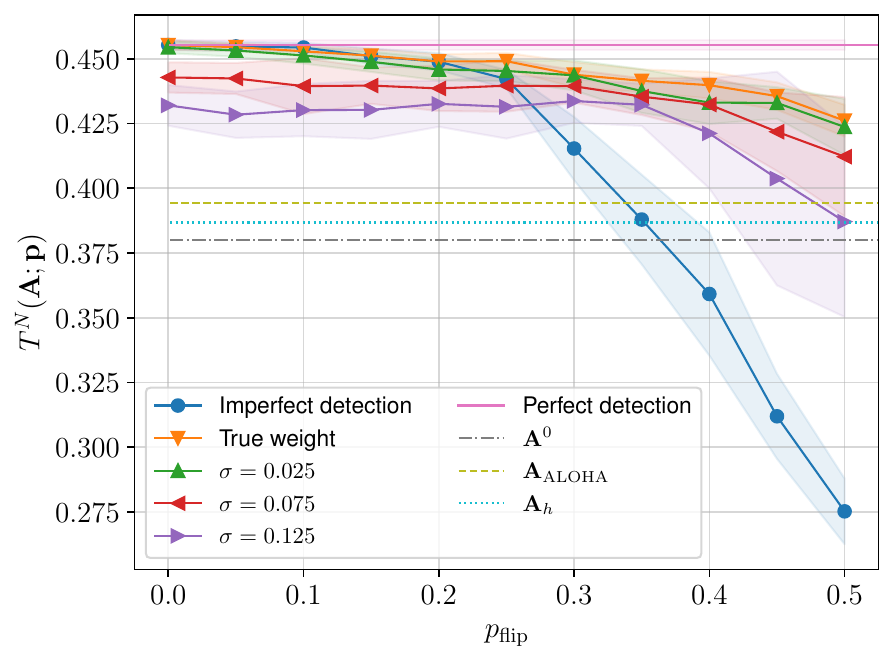}
\caption{Resulting throughput after 10000 frames for different values of $p_{\mathrm{flip}}$.}
    \label{fig:bsc}
\end{figure}

Fig.~\ref{fig:bsc} shows the impact of the error probability $p_{\mathrm{flip}}$ on the throughput achieved by each algorithm. The activity probability of each device $i$, $p_i$, is drawn once independently from the uniform distribution $\mathrm{Unif}[0,0.45]$ and kept constant across the different runs. 

Observe in Fig.~\ref{fig:bsc} that $\mathbf{A}_h$ (in dashed blue) has the worst performance. This is due to the fact that the activity probabilities are low and allocating only a single device to a slot is inefficient. The random initial allocation $\mathbf{A}_0$ and the ALOHA allocation $\mathbf{A}_{\mathrm{ALOHA}}$ are poorly performing. This is due to the fact that neither account for the throughput objective in (\ref{eq:T_i}). For small values of $p_\mathrm{flip}$ ($<0.25$) the algorithms have similar performance, however for higher error probability, applying Alg.~\ref{alg:SGA_no_bias} ignoring user identification errors (blue curve) leads to significant degradation in the throughput. The proposed Alg.~\ref{alg:SGA_bias_reduced} with perfect knowledge of the weight in (\ref{eq:weight}) (orange curve) and with imperfect weight (green, red, purple) always outperforms Alg.~\ref{alg:SGA_no_bias} with identification errors. This highlights the utility of our proposed algorithm as it mitigates the effect of user identification errors.
\begin{figure}
    \centering
    \includegraphics[width=\linewidth]{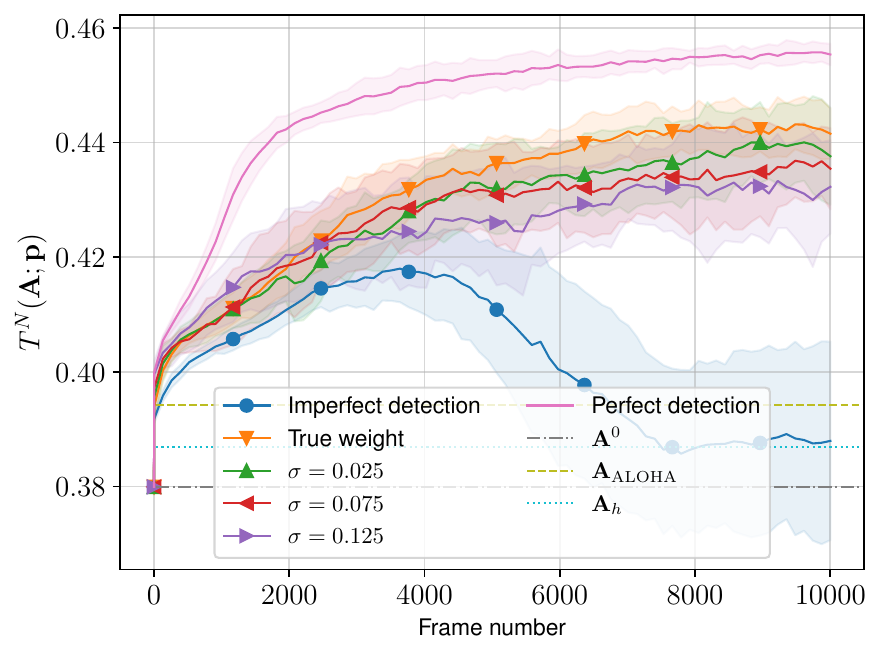}
    \caption{Trajectories of the different method presented in Fig.~\ref{fig:bsc} with $p_\mathrm{flip}=0.35$.}
    \label{fig:traj_bsc_0.15}
\end{figure}
Fig.~\ref{fig:traj_bsc_0.15} shows the trajectories for each method for   $p_\mathrm{flip}=0.35$. We can see that the blue curve (imperfect detection) starts by increasing slightly before falling into a bad local maxima whilst the proposed algorithms converge to a better one. We should note that the reason the curves for methods different than the perfect detection are not monotonically increasing is because they are evaluated based on the objective $T^N(\bf{A};\mathbf{p})$ for the true $\bf{p}$ whilst the activity vector $\hat{\bf{X}}$ used in the computation of the gradient is drawn from $\hat{\bf{p}}$.
The green, red, purple, and curves correspond to the proposed Alg.~\ref{alg:SGA_bias_reduced} with imperfect knowledge of the weight, obtained by perturbing the true weight by Gaussian noise. As the standard deviation of the noise increases, the throughput performance degrades. This is due to the fact that the noise introduces bias in the gradient estimates. Nevertheless, Alg.~\ref{alg:SGA_bias_reduced} with weights corrupted by noise still outperforms Alg.~\ref{alg:SGA_no_bias} (blue curve), which does not account for user identification errors.     

\subsection{Asymmetric Errors}
\label{sec:num:asym}
We now consider a detection algorithm that has $\mathrm{Pr}(\mathrm{false~alarm})=0$ and $\mathrm{Pr}(\mathrm{miss~detection})=p_\mathrm{miss}$ for $p_{\mathrm{miss}} \in [0,0.5]$ common to each device. In other words, if device $i$ is active, the probability that it is not detected is $p_{\mathrm{miss}}$. In contrast to the scenario of Fig.~\ref{fig:bsc}, if device $i$ is not active, then the probability of it being detected is zero.
\begin{align}
    \hat{p}_i = (1-p_{\mathrm{miss}})p_i,~i = 1,\ldots,N.
\end{align}
\begin{figure}[h!]
    \centering
    \includegraphics[width=\linewidth]{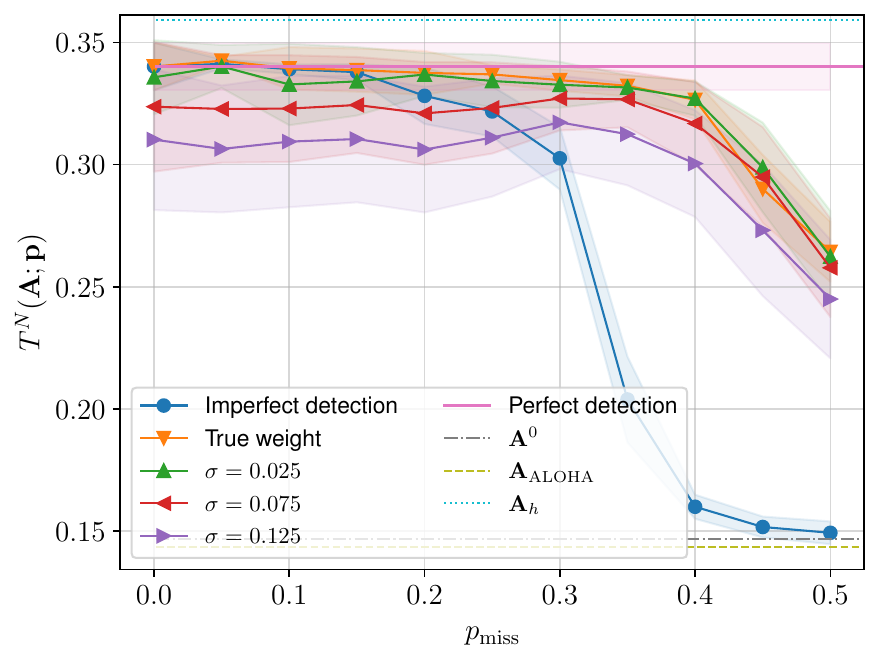}
    \caption{Resulting throughput after 10000 frames for different values of $p_{miss}$.}
    \label{fig:bec}
\end{figure}
Fig.~\ref{fig:bec} shows the impact of the error probability $p_{\mathrm{miss}}$ on the throughput achieved by each algorithm. The activity probability of each device $i$, $p_i$, is drawn independently from the uniform distribution $\mathrm{Unif}[0,1]$The activity probability of each device $i$, $p_i$, is drawn once independently from the uniform distribution $\mathrm{Unif}[0,0.9]$ and kept constant across the runs. That is, there is a high probability that many devices will have an activity probability $p_i > \frac{1}{2}$.

Observe in Fig.~\ref{fig:bec}, the initial random allocation $\mathbf{A}^0$ and the ALOHA allocation $\mathbf{A}_{\mathrm{ALOHA}}$ have the worst performance. Due to the high heterogeneity of the network, the gap between $\mathbf{A}_{\mathrm{ALOHA}}$ and the best performing method is larger than in the previous scenario. In this case, as some users have a probability of activity $p_i>\frac{1}{2}$, the greedy heuristic (in dashed blue) performs the best. 
The pink curve corresponds to Alg.~\ref{alg:SGA_no_bias} with perfect user activity estimation. For small and large values of $p_{\mathrm{miss}}$, the pink curve is higher than any of the curves corresponding to proposed Alg.~\ref{alg:SGA_bias_reduced} and Alg.~\ref{alg:SGA_no_bias} with user identification errors. In particular, observe that, as for the symmetric errors, applying Alg.~\ref{alg:SGA_no_bias} while ignoring user identification errors (blue curve) leads to significant degradation in the throughput. Alg.~\ref{alg:SGA_bias_reduced} with perfect knowledge of the weight in (\ref{eq:weight}) to mitigate user identification errors (orange curve) always outperforms Alg.~\ref{alg:SGA_no_bias}. 

Again, the green, red, purple, and curves show that the throughput degrades when the standard deviation of the prior increase but they still outperform Alg.~\ref{alg:SGA_no_bias} (blue curve), which does not account for user identification errors, when the number of errors is high.

\subsection{Errors Arising from GAMP-Based Detection}
\label{sec:num:gamp}
Finally, we consider the use of an active user detection algorithm based on GAMP \cite{Rangan2011generalized}. As described in Section~\ref{sec:preamble}, each user $i$ is given a unique complex preamble (or pilot) $\mathbf{s}_i$ of length $T=15$ with $s_{i,j}\sim\mathcal{CN}(0,1)$, for $j=1\ldots10$, in each frame a new set of pilots is drawn. As $T<N$, it is impossible to find a set pilots that are all mutually orthogonal. Pilots are sent with a unit transmit power and are received by a single antenna (i.e., $N_A = 1$ in this scenario). Users are active according to the following distribution: $\mathbf{p}=[0.01, 0.03, 0.09, 0.14, 0.21, 0.21, 0.23, 0.27, 0.32, 0.33, 0.34,\\ 0.42, 0.43, 0.47, 0.52, 0.56, 0.58, 0.61, 0.65, 0.8]$. Pilots are multiplied by a complex channel coefficient and corrupted by some Gaussian noise. The modulus of the channel coefficients are $|\mathbf{H}|=[1.6, 0.8, 0.5, 0.5, 1.2, 1., 2.4, 0.3, 1.0, 0.1, 0.5, 1.2, 1.7, 0.2,\\2.5, 1.6, 2.1, 1.4, 0.5, 0.2]$ (the matrix $\mathbf{H}$ is a vector as there is a single receive antenna).

As the posterior distribution $\hat{\mathbf{p}}$ of GAMP is difficult to compute it is estimated using a Monte Carlo approximation by running the network for $10000$ frames and this estimation is used as the denominator of the importance weight.  
\begin{figure}[h!]
    \includegraphics[width=\linewidth]{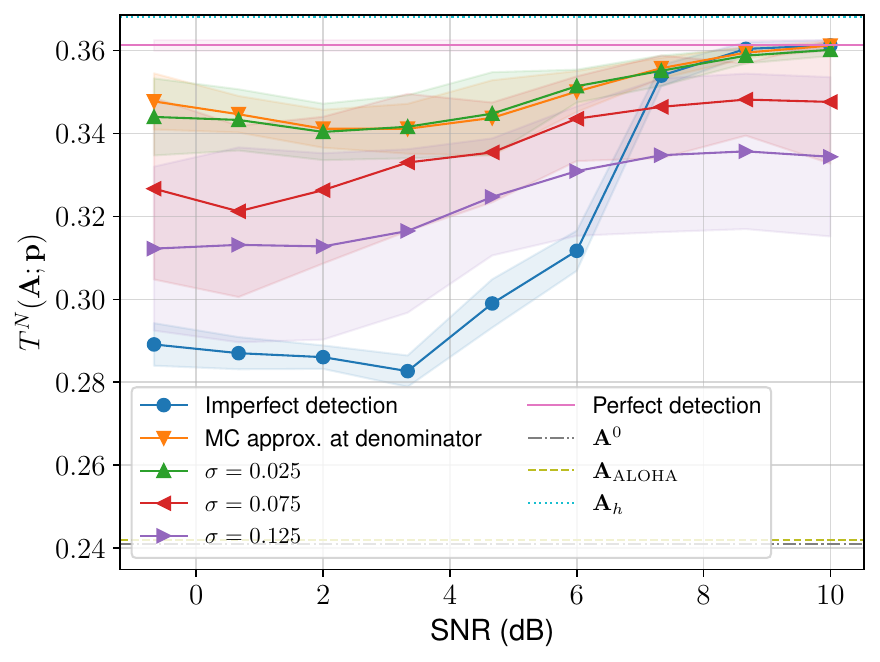}
    \caption{Resulting throughput after 10000 frames for different values of SNR.}
    \label{fig:gamp}
\end{figure}

Fig.~\ref{fig:gamp} shows the effect of the additive Gaussian noise on the throughput of the resulting allocation given by each algorithm. The figure is shown as a function de the transmit SNR, \textit{i.e.} the ratio of the transmit power, 1, over the variance of the noise, without taking into account the channel coefficient. Note that the received SNR of the user with the worst channel ranges between $[0,-10.6]$dB. The algorithms are run for 10000 frames. 

In Fig.~\ref{fig:gamp}, as for the asymmetric errors, the random initial allocation $\mathbf{A}^0$ and the ALOHA allocation $\mathbf{A}_{\mathrm{ALOHA}}$ have the worst performance. In this case, the greedy heuristic (dashed blue) performs the best as some users have a probability of activity $p_i>0.5$. 
For small and large values of SNR the orange curve representing the weight made of the true target and the empirical distribution estimation is above the curves of the algorithm of Alg.~\ref{alg:SGA_bias_reduced}.
For high values of SNR, almost no errors are introduced by the GAMP algorithm. Hence all algorithms achieve a similar throughput. For smaller values of SNR ($\leq 6$dB), the Alg.~\ref{alg:SGA_no_bias} with imperfect detection (blue curve) has degraded performance This is in parts due to the fact that the most likely device has a bad channel gain meaning that GAMP algorithm often fails to decoded it. Hence the resulting resource allocation does not manage to take it sufficiently into account leading to a decrease in performance. Note that the proposed Alg.~\ref{alg:SGA_bias_reduced_imperfect_w} (green, red, and purple) outperform Alg.~\ref{alg:SGA_no_bias} (blue curve). 

\section{Conclusions}\label{sec:conclusion}

A problem arising in the design of frame slotted ALOHA algorithms with heterogeneous devices is the need for perfect activity estimates in each frame. In this paper, we have studied the impact of imperfect activity estimation and observed a significant throughput performance degradation in existing algorithms. To overcome this problem, we have proposed an importance weighted bias mitigation strategy. This algorithm is shown to be capable of guaranteeing almost sure convergence of stochastic gradient methods to a stationary point of the throughput maximization problem. A numerical study shows that our method outperforms existing algorithms which do not account for imperfect activity estimation. Moreover, our approach is robust to uncertainty in the importance weights. 

\appendices

\section{Proof of Theorem~\ref{thrm:convergence}}\label{app:proof}

	Observe that the update rule can be written as 
\begin{align}
    \mathbf{A}^{t+1} &= \mathbf{A}^t + \gamma^{t+1}\left(\mathbb{E}_{\hat{\mathbf{X}}^{t+1}\sim\hat{\mathbf{p}}}[w(\hat{\mathbf{X}}^{t+1})g(\mathbf{A}^t;\hat{\mathbf{X}}^{t+1})]\right.\notag\\
                     & \left. w(\hat{\mathbf{X}}^{t+1})g(\mathbf{A}^{t};\hat{\mathbf{X}}^{t+1}) - \mathbb{E}_{\hat{\mathbf{X}}^{t+1}\sim\hat{\mathbf{p}}}[w(\hat{\mathbf{X}}^{t+1})g(\mathbf{A}^t;\hat{\mathbf{X}}^{t+1})]\right)\notag\\
	&~~~ + \gamma^{t+1}\mathbf{Z}^{t+1}\notag\\
	&= \mathbf{A}^{t} + \gamma^{t+1}\left(\mathbb{E}_{\mathbf{X}^{t+1}\sim\mathbf{p}}[g(\mathbf{A}^t;\mathbf{X}^{t+1})] + \partial \tilde{\mathbf{M}}^{t+1}\right)\notag\\
	&~~~+ \gamma^{t+1}\mathbf{Z}^{t+1}\notag\\
    &= \mathbf{A}^{t+1} + \gamma^{t+1}\left(\nabla_{\mathbf{A}^t} T(\mathbf{A}^t) + \partial \tilde{\mathbf{M}}^{t+1}\right) + \gamma^{t+1}\mathbf{Z}^{t+1},
\end{align}
where
\begin{align}
	&\partial \tilde{\mathbf{M}}^{t+1}\notag\\
    &= w(\hat{\mathbf{X}}^{t+1})g(\mathbf{A}^{t};\hat{\mathbf{X}}^{t+1}) - \mathbb{E}_{\hat{\mathbf{X}}^{t+1}\sim\hat{\mathbf{p}}}[w(\hat{\mathbf{X}}^{t+1})g(\mathbf{A}^t;\hat{\mathbf{X}}^{t+1})]
\end{align} 
satisfies $\mathbb{E}_{t+1}[\partial \tilde{\mathbf{M}}^{t+1}] = 0$, where $\mathbb{E}_{t+1}[\cdot]$ is the conditional expectation on with respect to the $\sigma$-algebra generated by the iterates up to time $t$. The last equality then follows from the identity in (\ref{eq:IS_identity}).

To establish convergence, we apply Theorem 5.2.1 in \cite{kushnerStochasticApproximationRecursive2003}. To do so, we verify the following conditions:
\begin{itemize}
	\item A.5.2.1: 
	\begin{align}
        \sup_t \mathbb{E}_{\hat{\mathbf{X}}^{t+1}\sim\hat{\mathbf{p}}}[w(\hat{\mathbf{X}}^{t+1})g(\mathbf{A}^t;\hat{\mathbf{X}}^{t+1})] < \infty,
	\end{align}
	which holds under the assumption that $w(\mathbf{x}) < \infty$ for all $\mathbf{x} \in \{0,1\}^N$.
	\item A.5.2.2: There exists $\overline{g}: \mathbb{R}^{N \times M} \rightarrow \mathbb{R}$ such that 
	\begin{align}
        \mathbb{E}_{\hat{\mathbf{X}}^{t+1}\sim\hat{\mathbf{p}}}[w(\hat{\mathbf{X}}^{t+1})g(\mathbf{A}^t;\hat{\mathbf{X}}^{t+1})] = \overline{g}(\mathbf{A}^t),
	\end{align} 
    which holds by setting $\overline{g}(\mathbf{A}^t) = \nabla_{\mathbf{A}^t} T(\mathbf{A}^t; \mathbf{p})$.
\item A.5.2.3: $\nabla_{\mathbf{A}} T(\mathbf{A};\mathbf{p})$ is continuous, which immediately follows from the definition of $T(\mathbf{A};\mathbf{p})$. 
	\item A.5.2.5: The gradient estimate $w(\hat{\mathbf{X}}^{t+1})g(\mathbf{A}^t;\hat{\mathbf{X}}^{t+1})$ is unbiased, which follows from (\ref{eq:IS_identity}). 
\end{itemize}
A further condition on the constraint must also hold. As the same constraint is present in \cite{zhengStochasticResourceOptimization2021} for the case without errors, the same argument can be applied to complete the proof. 
 
\bibliographystyle{IEEEtran}
\bibliography{Aloha_new.bib}

\begin{thebibliography}{10}
\providecommand{\url}[1]{#1}
\csname url@samestyle\endcsname
\providecommand{\newblock}{\relax}
\providecommand{\bibinfo}[2]{#2}
\providecommand{\BIBentrySTDinterwordspacing}{\spaceskip=0pt\relax}
\providecommand{\BIBentryALTinterwordstretchfactor}{4}
\providecommand{\BIBentryALTinterwordspacing}{\spaceskip=\fontdimen2\font plus
\BIBentryALTinterwordstretchfactor\fontdimen3\font minus
  \fontdimen4\font\relax}
\providecommand{\BIBforeignlanguage}[2]{{%
\expandafter\ifx\csname l@#1\endcsname\relax
\typeout{** WARNING: IEEEtran.bst: No hyphenation pattern has been}%
\typeout{** loaded for the language `#1'. Using the pattern for}%
\typeout{** the default language instead.}%
\else
\language=\csname l@#1\endcsname
\fi
#2}}
\providecommand{\BIBdecl}{\relax}
\BIBdecl

\bibitem{vtc}
A.~Jeannerot, M.~Egan, L.~Chetot, and J.-M. Gorce, ``Mitigating {{User
  Identification Errors}} in {{Resource Optimization}} for {{Grant-Free Random
  Access}},'' in \emph{2023 {{IEEE}} 97th {{Vehicular Technology Conference}}
  ({{VTC2023-Spring}})}, Jun. 2023, pp. 1--6.

\bibitem{Shahab2020grant}
M.~Shahab \emph{et~al.}, ``Grant-free non-orthogonal multiple access for the
  {IoT}: a survey,'' \emph{IEEE Communications Surveys \& Tutorials}, vol.~22,
  no.~3, pp. 1805--1838, 2020.

\bibitem{Liu2018sparse}
L.~Liu \emph{et~al.}, ``Sparse signal prcoessing for grant-free massive
  connectivity: a future paradigm for random access protocols in the {I}nternet
  of {T}hings,'' \emph{IEEE Signal Processing Magazine}, vol.~35, no.~5, pp.
  88--99, 2018.

\bibitem{Choi2021grant}
J.~Choi \emph{et~al.}, ``Grant-free random access in machine-type
  communication: approaches and challenges,'' \emph{IEEE Wireless
  Communications}, vol.~29, no.~1, pp. 151--158, 2021.

\bibitem{Ding2017survey}
Z.~Ding, X.~Lei, G.~Karagiannidis, R.~Schober, J.~Yuan, and V.~Bhargava, ``A
  survey on non-orthogonal multiple access for {5G} networks: research
  challenges and future trends,'' \emph{IEEE Journal on Selected Areas in
  Communications}, vol.~35, no.~10, pp. 2181--2195, 2017.

\bibitem{robertsALOHAPacketSystem1975}
L.~Roberts, ``{ALOHA} packet system with and without slots and capture,''
  \emph{SIGCOMM Comput. Commun. Rev.}, vol.~5, no.~2, pp. 28--42, 1975.

\bibitem{wieselthierExactAnalysisPerformance1989b}
J.~Wieselthier, A.~Ephremides, and L.~Michaels, ``An exact analysis and
  performance evaluation of framed {ALOHA} with capture,'' \emph{IEEE
  Transactions on Communications}, vol.~37, no.~2, pp. 125--137, 1989.

\bibitem{livaGraphBasedAnalysisOptimization2011}
G.~Liva, ``Graphi-based analysis and optimization of contention resolution
  diversity slotted {ALOHA},'' \emph{IEEE Transactions on Communications},
  vol.~59, no.~2, pp. 477--487, 2011.

\bibitem{paoliniCodedSlottedALOHA2015}
E.~Paolini, G.~Liva, and M.~Chiani, ``Coded slotted {ALOHA}: a graph-based
  method for uncoordinated multiple access,'' \emph{IEEE Transactions on
  Information Theory}, vol.~61, no.~12, pp. 6815--6832, 2015.

\bibitem{Yu2020stabilizing}
J.~Yu \emph{et~al.}, ``Stablizing frame slotted {Aloha}-based {IoT} systems: a
  geometric ergodicity perspective,'' \emph{IEEE Journal on Selected Areas in
  Communications}, vol.~39, no.~3, pp. 714--725, 2020.

\bibitem{Yue2023age}
Z.~Yue \emph{et~al.}, ``Age of information under frame slotted {ALOHA}-based
  status updating protocol,'' \emph{IEEE Journal on Selected Areas in
  Communications}, 2023.

\bibitem{Huang2023age}
Y.~Huang \emph{et~al.}, ``Age of information minimization for frameless {ALOHA}
  in grant-free massive access,'' \emph{IEEE Transactions on Wireless
  Communications}, 2023.

\bibitem{casiniContentionResolutionDiversity2007}
E.~Casini, R.~De~Gaudenzi, and O.~Del Rio~Herrero, ``Contention resolution
  diversity slotted {ALOHA} ({CRDSA}): an enhanced random access scheme for
  satellite access packet networks,'' \emph{IEEE Transactions on Wireless
  Communications}, vol.~6, no.~4, pp. 1408--1419, 2007.

\bibitem{verdu_multiuser_1998}
S.~Verd{\'u}, \emph{Multiuser Detection}.\hskip 1em plus 0.5em minus
  0.4em\relax New York, NY: Cambridge University Press, 1998.

\bibitem{gallagerLowDensityParityCheckCodes1962}
R.~Gallager, ``Low-density parity-check codes,'' \emph{IRE Transactions on
  Information Theory}, vol.~8, no.~1, pp. 21--28, 1962.

\bibitem{Shieh2022enhanced}
S.-L. Shieh and S.-H. Yang, ``Enhanced irregular repition slotted {ALOHA} under
  {SIC} limitation,'' \emph{IEEE Transactions on Communications}, vol.~70,
  no.~4, pp. 2268--2280, 2022.

\bibitem{Srivatsa2022user}
C.~Srivatsa and C.~Murthy, ``User activity detection for irregular repetition
  slotted {ALOHA} based {MMTC},'' \emph{IEEE Transactions on Signal
  Processing}, vol.~70, pp. 3616--3631, 2022.

\bibitem{Zhang2022performance}
Z.~Zhang, K.~Niu, and J.~Dai, ``Performance bounds of coded slotted {ALOHA}
  over erasure channels,'' \emph{IEEE Transactions on Vehicular Technology},
  vol.~71, no.~11, pp. 12\,338--12\,343, 2022.

\bibitem{Dumas2021design}
C.~Dumas \emph{et~al.}, ``Design of coded slotted {ALOHA} with interference
  cancellation errors,'' \emph{IEEE Transactions on Vehicular Technology},
  vol.~70, no.~12, pp. 12\,742--12\,757, 2021.

\bibitem{Su2021noma}
J.~Su, G.~Ren, and B.~Zhao, ``{NOMA}-based coded slotted {ALOHA} for
  machine-type communications,'' \emph{IEEE Communications Letters}, vol.~25,
  no.~7, pp. 2435--2439, 2021.

\bibitem{Sousa2023study}
M.~Sousa-Vieira and M.~Fern{\`a}ndez-Veiga, ``Study of coded {ALOHA} with
  multi-user detection under heavy-tailed and correlated arrivals,''
  \emph{Future Internet}, vol.~15, no.~4, p. 132, 2023.

\bibitem{kalorRandomAccessSchemes2018a}
A.~E. Kalor, O.~A. Hanna, and P.~Popovski, ``Random access schemes in wireless
  systems with correlated user activity,'' in \emph{2018 IEEE 19th
  International Workshop on Signal Processing Advances in Wireless
  Communications (SPAWC)}, 2018, pp. 1--5.

\bibitem{zhengStochasticResourceOptimization2021}
C.~Zheng, M.~Egan, L.~Clavier, A.~E. Kalør, and P.~Popovski, ``Stochastic
  resource optimization of random access for transmitters with correlated
  activation,'' \emph{IEEE Communications Letters}, vol.~25, no.~9, pp.
  3055--3059, 2021.

\bibitem{zhengStochasticResourceAllocation2022}
------, ``Stochastic resource allocation for outage minimization in random
  access with correlated activation,'' in \emph{2022 IEEE Wireless
  Communications and Networking Conference (WCNC)}, 2022, pp. 1635--1640.

\bibitem{Ke2020compressive}
M.~Ke, Z.~Gao, Y.~Wu, X.~Gao, and R.~Schober, ``Compressive sensing-based
  adaptive active user detection and channel estimation: massive access meets
  massive {MIMO},'' \emph{IEEE Transactions on Signal Processing}, vol.~68, pp.
  764--779, 2020.

\bibitem{Zou2020message}
Q.~Zou, H.~Zhang, D.~Cai, and H.~Yang, ``Message passing based joint channel
  and user activity estimation for uplink grant-free massive {MIMO} systems
  with low-precision {ADCs},'' \emph{IEEE Signal Processing Letters}, vol.~27,
  pp. 506--510, 2020.

\bibitem{Chetot2021joint}
L.~Chetot, M.~Egan, and J.-M. Gorce, ``Joint identification and channel
  estimation for fault detection in industrial {IoT} with correlated sensors,''
  \emph{IEEE Access}, vol.~9, pp. 116\,692--116\,701, 2021.

\bibitem{Chetot2023active}
------, ``Active user detection and channel estimation for grant-free random
  access with {G}aussian correlated activity,'' in \emph{IEEE Vehicular
  Technology Conference (VTC-Spring)}, 2023.

\bibitem{panAIDrivenBlindSignature2022}
J.~Pan, N.~Ye, H.~Yu, T.~Hong, S.~Al-Rubaye, S.~Mumtaz, A.~Al-Dulaimi, and
  I.~Chih-Lin, ``{{AI-Driven Blind Signature Classification}} for {{IoT
  Connectivity}}: {{A Deep Learning Approach}},'' \emph{IEEE Transactions on
  Wireless Communications}, vol.~21, no.~8, pp. 6033--6047, 2022.

\bibitem{zadroznyLearningEvaluatingClassifiers2004}
\BIBentryALTinterwordspacing
B.~Zadrozny, ``Learning and evaluating classifiers under sample selection
  bias,'' in \emph{Proceedings of the Twenty-First International Conference on
  Machine Learning}, ser. ICML '04.\hskip 1em plus 0.5em minus 0.4em\relax New
  York, NY, USA: Association for Computing Machinery, 2004, p. 114. [Online].
  Available: \url{https://doi.org/10.1145/1015330.1015425}
\BIBentrySTDinterwordspacing

\bibitem{ghalebikesabiMitigatingStatisticalBias2022}
\BIBentryALTinterwordspacing
S.~Ghalebikesabi, H.~Wilde, J.~Jewson, A.~Doucet, S.~Vollmer, and C.~Holmes,
  ``Mitigating statistical bias within differentially private synthetic data,''
  in \emph{Proceedings of the Thirty-Eighth Conference on Uncertainty in
  Artificial Intelligence}, ser. Proceedings of Machine Learning Research,
  J.~Cussens and K.~Zhang, Eds., vol. 180.\hskip 1em plus 0.5em minus
  0.4em\relax PMLR, 01--05 Aug 2022, pp. 696--705. [Online]. Available:
  \url{https://proceedings.mlr.press/v180/ghalebikesabi22a.html}
\BIBentrySTDinterwordspacing

\bibitem{hoglundOverview3GPPRelease2017}
A.~Hoglund, X.~Lin, O.~Liberg, A.~Behravan, E.~Yavuz, M.~Van Der~Zee, Y.~Sui,
  T.~Tirronen, A.~Ratilainen, and D.~Eriksson, ``Overview of {3GPP} release 14
  enhanced {NB-IoT},'' \emph{IEEE Network}, vol.~31, no.~6, pp. 16--22, 2017.

\bibitem{veeduSmallerLowerCost5G2022}
S.~Veedu \emph{et~al.}, ``Toward smaller and lower-cost {5G} devices with
  longer battery life: an overview of {3GPP} release 17 {Redcap},'' \emph{IEEE
  Communications Standards Magazine}, vol.~6, no.~3, pp. 84--90, 2022.

\bibitem{miaoNarrowbandInternetThings2018}
Y.~Miao, W.~Li, D.~Tian, M.~S. Hossain, and M.~F. Alhamid, ``Narrowband
  {{Internet}} of {{Things}}: {{Simulation}} and {{Modeling}},'' \emph{IEEE
  Internet of Things Journal}, vol.~5, no.~4, pp. 2304--2314, Aug. 2018.

\bibitem{kushnerStochasticApproximationRecursive2003}
H.~Kushner and G.~Yin, \emph{Stochastic Approximation and Recursive Algorithms
  and Applications}.\hskip 1em plus 0.5em minus 0.4em\relax New York, NY:
  Springer, 2003.

\bibitem{Rangan2011generalized}
S.~Rangan, ``Generalized approximate message passing for estimation with random
  linear mixing,'' in \emph{2011 IEEE International Symposium on Information
  Theory Proceedings}, 2011, pp. 2168--2172.

\end{thebibliography}
\end{document}